\documentclass[conference]{IEEEtran}
\usepackage{cite}
\usepackage[pdftex]{graphicx}
\usepackage[cmex10]{amsmath}
\newtheorem{definition}{Definition} 
\newtheorem{theorem}{Theorem} 

\usepackage{algorithmic}
\usepackage{algorithm}
\newtheorem{proposition}{Proposition} 

\usepackage{algorithmic}
\begin{document}
\title{Large Social Networks can be Targeted for Viral Marketing with Small Seed Sets}

\author{\IEEEauthorblockN{Paulo Shakarian and Damon Paulo}
\IEEEauthorblockA{Network Science Center and \\Department of Electrical Engineering and Computer Science\\
United States Military Academy\\
West Point, New York 10996\\
Email: paulo[at]shakarian.net, damon.paulo[at]usma.edu}
}
\maketitle

\begin{abstract}
In a ``tipping'' model, each node in a social network, representing an individual, adopts a behavior if a certain number of his incoming neighbors previously held that property.  A key problem for viral marketers is to determine an initial ``seed'' set in a network such that if given a property then the entire network adopts the behavior.  Here we introduce a method for quickly finding seed sets that scales to very large networks.  Our approach finds a set of nodes that guarantees spreading to the entire network under the tipping model.  After experimentally evaluating $31$ real-world networks, we found that our approach often finds such sets that are several orders of magnitude smaller than the population size.  Our approach also scales well - on a Friendster social network consisting of $5.6$ million nodes and $28$ million edges we found a seed sets in under $3.6$ hours.  We also find that highly clustered local neighborhoods and dense network-wide community structure together suppress the ability of a trend to spread under the tipping model.
\end{abstract}

\IEEEpeerreviewmaketitle

\section{Introduction}
A much studied model in network science, tipping\cite{Gran78,Schelling78,jy05} (a.k.a. deterministic linear threshold\cite{kleinberg}) is often associated with ``seed'' or ``target'' set selection,~\cite{chen09siam} (a.k.a. the maximum influence problem).  In this problem we have a social network in the form of a directed graph and thresholds for each individual.  Based on this data, the desired output is the smallest possible set of individuals such that, if initially activated, the entire population will adopt the new behavior (a seed set).  This problem is NP-Complete~\cite{kleinberg,Dreyer09}.  Although approximation algorithms have been proposed,~\cite{leskovec07,chen09siam,benzwi09,chen10} none seem to scale to very large data sets.  Here, inspired by shell decomposition,~\cite{ShaiCarmi07032007,InfluentialSpreaders_2010,baxter11} we present a method guaranteed to find a set of nodes that causes the entire population to activate - but is not necessarily of minimal size.  We then evaluate the algorithm on $31$ large real-world social networks and show that it often finds very small seed sets (often several orders of magnitude smaller than the population size).  We also show that the size of a seed set is related to Louvain modularity and average clustering coefficient.  Therefore, we find that dense community structure and tight-knit local neighborhoods together inhibit the spreading of trends under the tipping model.

The rest of the paper is organized as follows.  In Section~\ref{prelim-sec}, we provide formal definitions of the tipping model.  This is followed by the presentation of our new algorithm in Section~\ref{alg-sec}.  We then describe our experimental results in Section~\ref{res-sec}.  Finally, we provide an overview of related work in Section~\ref{rw-sec}.
\section{Technical Preliminaries}
\label{prelim-sec}
Throughout this paper we assume the existence of a \textit{social network,} $G=(V,E)$, where $V$ is a set of vertices and $E$ is a set of directed edges.  We will use the notation $n$ and $m$ for the cardinality of $V$ and $E$ respectively.  For a given node $v_i \in V$, the set of incoming neighbors is $\eta^{in}_i$, and the set of outgoing neighbors is $\eta^{out}_i$.  The cardinalities of these sets (and hence the in and out degrees of node $v_i$) are $d^{in}_i, d^{out}_i$ respectively.  We now define a threshold function that for each node returns the fraction of incoming neighbors that must be activated for it to become activate as well.

\begin{definition}[Threshold Function]
We define the \textbf{threshold function} as mapping from V to $ (0,1] $.  Formally: $ \theta: V \rightarrow (0,1] $.
\end{definition}

For the number of neighbors that must be active, we will use the shorthand $k_i$.  Hence, for each $v_i$, $k_i = \lceil \theta (v_i) \cdot d^{in}_i \rceil$.  We now define an \textit{activation function} that, given an initial set of active nodes, returns a set of active nodes after one time step.

\begin{definition}[Activation Function]
Given a threshold function, $ \theta $, an \textbf{activation function} $ A_{\theta} $ maps subsets of V to subsets of V, 
where for some $ V' \subseteq V $,
\begin{equation}
A_{\theta}(V') = V' \cup \{ v_i \in V\ s.t.\ |\eta^{in}_i \cap V'| \geq k_i \} 
\end{equation}
\end{definition}

We now define multiple applications of the activation function.

\begin{definition}[Multiple Applications of the Activation Function]
Given a natural number $ i > 0$, set $V' \subseteq V $, and threshold function, $ \theta $, we define the multiple applications of the activation function, ${A^i_{\theta}}(V')$, as follows:
\begin{equation}
 A^i_\theta(V') = \begin{cases} A_\theta(V' ) & \text{if $i=1$}  \\ A_\theta(A^{i-1}_\theta(V' )) & \text{otherwise} \end{cases} 
\end{equation}
\end{definition}

Clearly, when $ A^i_\theta(V')= A^{i-1}_\theta(V')$  the process has converged.  Further, this occurs in no more than $n$ steps (as, in each step, at least one new node must be activated).  Based on this idea, we define the function $\Gamma$ which returns the set of all nodes activated upon the convergence of the activation function.  

\begin{definition}[$\Gamma$ Function]
Let j be the least value such that $ A^j_{\theta}(V') = A^{j-1}_{\theta}(V') $.  We define the function $\Gamma_\theta : 2^V \rightarrow 2^V$ as follows.
\begin{equation}
\mathbf{ \Gamma_\theta } (V') = A^j_{ \theta }(V')
\end{equation}
\end{definition}

We now have all the pieces to introduce our problem - finding the minimal number of nodes that are initially active to ensure that the entire set $V$ becomes active.

\begin{definition}[The MIN-SEED Problem]
The MIN-SEED Problem is defined as follows: given a threshold function, $ \theta $, return $ V' \subseteq V\ s.t.\
\Gamma_\theta (V') = V $, and there does not exist $ V'' \subseteq V $ where $ |V''| < |V'| $ and 
$ \Gamma_\theta (V'') = V $.
\end{definition}

The following theorem is from the literature~\cite{kleinberg,Dreyer09} and tells us that the MIN-SEED problem is NP-complete.

\begin{theorem}[Complexity of MIN-SEED~\cite{kleinberg,Dreyer09}]
MIN-SEED in NP-Complete.
\end{theorem}

\section{Algorithm}
\label{alg-sec}
To deal with the intractability of the MIN-SEED problem, we design an algorithm that finds a non-trivial subset of nodes that causes the entire graph to  activate, but we do not guarantee that the resulting set will be of minimal size.  The algorithm is based on the idea of shell decomposition often cited in physics literature~\cite{Seidman83,ShaiCarmi07032007,InfluentialSpreaders_2010,baxter11} but modified to ensure that the resulting set will lead to all nodes being activated.  The algorithm, \textsf{TIP\_DECOMP} is presented in this section.

\algsetup{indent=1em}
	\begin{algorithm}[h!]
		\caption{ \textsf{TIP\_DECOMP}}
		\begin{algorithmic}[1]

		\REQUIRE Threshold function, $ \theta $ and directed social network $G=(V,E)$
		\ENSURE $ V' $
		\medskip

		\STATE{ For each vertex $ v_i $, compute $ k_i $}.
		\STATE{ For each vertex $ v_i,\ dist_i = d_i^{in} - k_i $}.
		\STATE{ FLAG = TRUE}.
		\WHILE{ FLAG }
			\STATE {Let $ v_i $ be the element of $ v $ where $ dist_i $ is minimal}.
			\IF { $ dist_i = \infty $}
				\STATE{ FLAG = FALSE}.
			\ELSE
				\STATE{ Remove $ v_i $ from $G$ and for each $ v_j $ in $ \eta_i^{out} $, if $dist_j > 0$, set
				$  dist_j = dist_j-1 $.  Otherwise set $dist_j = \infty $}.
			\ENDIF
		\ENDWHILE
		\RETURN{ All nodes left in $ G $}.
	\end{algorithmic}
\end{algorithm}

%
%
%

Intuitively, the algorithm proceeds as follows (Figure 1).  Given network $G=(V,E)$ where each node $v_i$ has threshold $k_i = \lceil \theta (v_i) \cdot d^{in}_i \rceil$, at each iteration, pick the node for which $d^{in}_i - k_i$ is the least but positive (or $0$) and remove it.  Once there are no nodes for which $d^{in}_i - k_i$ is positive (or $0$), the algorithm outputs the remaining nodes in the network.  

\begin{figure}
    \begin{center}
        \includegraphics[width=1\linewidth]{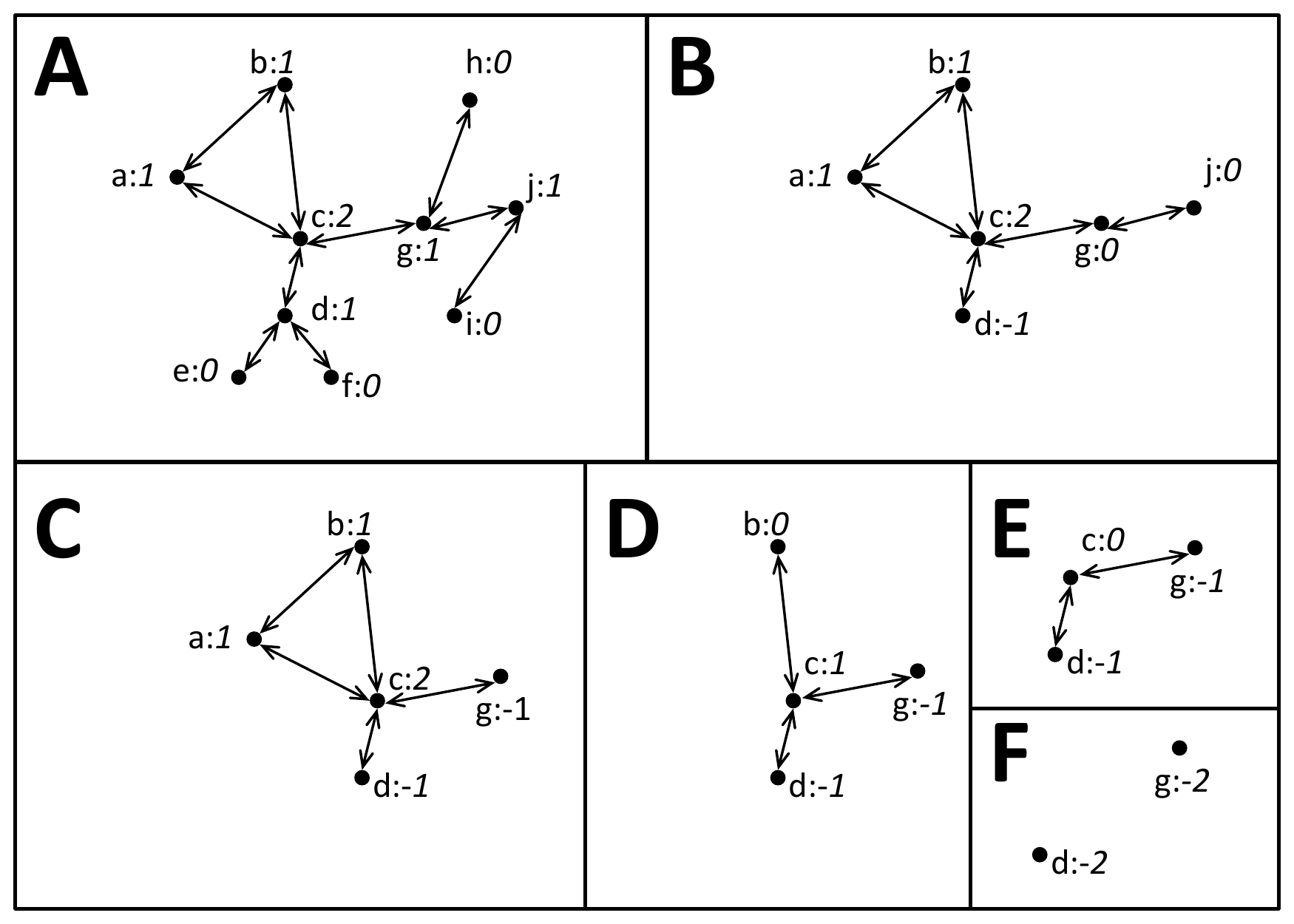}
    \end{center}
\caption{Example of our algorithm for a simple network depicted in box \textbf{A}.  We use a threshold value set to $50\%$ of the node degree.  Next to each node label (lower-case letter) is the value for $d^{in}_i - k_i$ (where $k_i = \lceil \frac{d^{in}_i}{2} \rceil$).  In the first four iterations, nodes e, f, h, and i are removed resulting in the network in box \textbf{B}.  This is followed by the removal of node j resulting in the network in box \textbf{C}.  In the next two iterations, nodes a and b are removed (boxes \textbf{D}-\textbf{E} respectively).  Finally, node c is removed (box \textbf{F}).  The nodes of the final network, consisting of d and g, have negetive values for $d_i-\theta_i$ and become the output of the algorithm.}
\end{figure}

Now, we prove that the resulting set of nodes is guaranteed to cause all nodes in the graph to activate under the tipping model.  This proof follows from the fact that any node removed is activated by the remaining nodes in the network.  

\begin{theorem}
If all nodes in $ V'\ \subseteq\ V $ returned by \textsf{TIP\_DECOMP} are initially active, then every node in $ V $ will eventually be activated, too. 
\end{theorem}

\begin{proof}
Let $ w $ be the total number of nodes removed by \textsf{TIP\_DECOMP}, where $ v_1 $ is the last node removed and $ v_w $ is the first node removed.  We prove the theorem by induction on $ w $ as follows.  We use $P(w)$ to denote the inductive hypothesis which states that all nodes from $ v_1 $ to $ v_w $ are active. In the base case, $P(1)$ trivially holds as we are guaranteed that from set $V'$ there are at least $k_1$ edges to $v_1$ (or it would not be removed).  For the inductive step, assuming $P(w)$ is true, when $ v_{w+1} $ was removed from the graph $ dist_{w+1} \geq 0 $ which means that $ d_{w+1}^{in} \geq k_{w+1} $.  All nodes in $\eta^{in}_{w+1}$ at the time when $ v_{w+1} $ was removed are now active, so $ v_{w+1} $ will now be activated - which completes the proof.
\end{proof}

We also note that by using the appropriate data structure (we used a binomial heap in our implementation), for a network of $n$ nodes and $m$ edges, this algorithm can run in time $O(m \log n)$.

\begin{proposition}
\label{tcomp}
The complexity of \textsf{TIP\_DECOMP} is $ O(m \cdot log(n)) $.
\end{proposition}
\section{Results}
\label{res-sec}
All experiments were run on a computer equipped with an Intel X5677 Xeon Processor operating at 3.46 GHz with a 12 MB Cache.  The machine was running Red Hat Enterprise Linux version 6.1 and equipped with 70 GB of physical memory.  \textsf{TIP\_DECOMP} was written using Python 2.6.6 in 200 lines of code that leveraged the NetworkX library available from http://networkx.lanl.gov/.  The code used a binomial heap library written by Bj\"orn B. Brandenburg available from http://www.cs.unc.edu/$\sim$bbb/.  All statistics presented in this section were calculated using R 2.13.1.

\subsection{Datasets}

In total, we examined $31$ networks: nine academic collaboration networks, three e-mail networks, and $19$ networks extracted from social-media sites.  The sites included included general-purpose social-media (similar to Facebook or MySpace) as well as special-purpose sites (i.e. focused on sharing of blogs, photos, or video).

All datasets used in this paper were obtained from one of four sources: the ASU Social Computing Data Repository,~\cite{Zafarani+Liu:2009} the Stanford Network Analysis Project,~\cite{snap} the University of Michigan,~\cite{umich} and Universitat Rovira i Virgili.\cite{uvi}  All networks considered were symmetric -- i.e. if a directed edge from vertex $v$ to $v'$ exists, there is also an edge from vertex $v'$ to $v$.  Tables~\ref{fig3} (A-C) show some of the pertinent qualities of these networks.  The networks are categorized by the results (explained later in this section).  In what follows, we provide their real-world context.

\subsection{Category A}
\begin{itemize}
\item{\textbf{BlogCatalog} is a social blog directory that allows users to share blogs with friends.~\cite{Zafarani+Liu:2009}  The first two samples of this site, BlogCatalog1 and 2, were taken in Jul. 2009 and June 2010 respectively.  The third sample, BlogCatalog3 was uploaded to ASU's Social Computing Data Repository in Aug. 2010.}
\item{\textbf{Buzznet} is a social media network designed for sharing photographs, journals, and videos.~\cite{Zafarani+Liu:2009}  It was extracted in Nov. 2010.}
\item{\textbf{Douban } is a Chinese social medial website designed to provide user reviews and recommendations.~\cite{Zafarani+Liu:2009}  It was extracted in Dec. 2010.}
\item{\textbf{Flickr} is a social media website that allows users to share photographs.~\cite{Zafarani+Liu:2009}  It was uploaded to ASU's Social Computing Data Repository in Aug. 2010.}
\item{\textbf{Flixster} is a social media website that allows users to share reviews and other information about cinema.~\cite{Zafarani+Liu:2009}  It was extracted in Dec. 2010.}
\item{\textbf{FourSquare} is a location-based social media site.~\cite{Zafarani+Liu:2009}  It was extracted in Dec. 2010.}
\item{\textbf{Frienster} is a general-purpose social-networking site.~\cite{Zafarani+Liu:2009} It was extracted in Nov. 2010.}
\item{\textbf{Last.Fm} is a music-centered social media site.~\cite{Zafarani+Liu:2009} It was extracted in Dec. 2010.}
\item{\textbf{LiveJournal} is a site designed to allow users to share their blogs.~\cite{Zafarani+Liu:2009}  It was extracted in Jul. 2010.}
\item{\textbf{Livemocha} is touted as the ``world's largest language community.''~\cite{Zafarani+Liu:2009}  It was extracted in Dec. 2010.}
\item{\textbf{WikiTalk} is a network of individuals who set and received messages while editing WikiPedia pages.~\cite{snap}  It was extracted in Jan. 2008.}
\end{itemize}

%

\subsection{Category B}
\begin{itemize}
\item{\textbf{Delicious} is a social bookmarking site, designed to allow users to share web bookmarks with their friends.~\cite{Zafarani+Liu:2009} It was extracted in Dec. 2010.}
\item{\textbf{Digg} is a social news website that allows users to share stories with friends.~\cite{Zafarani+Liu:2009}  It was extracted in Dec. 2010.}
\item{\textbf{EU E-Mail} is an e-mail network extracted from a large European Union research institution.~\cite{snap}  It is based on e-mail traffic from Oct. 2003 to May 2005.}
\item{\textbf{Hyves} is a popular general-purpose Dutch social networking site.~\cite{Zafarani+Liu:2009}  It was extracted in Dec. 2010.}
\item{\textbf{Yelp} is a social networking site that allows users to share product reviews.~\cite{Zafarani+Liu:2009}  It was extracted in Nov. 2010.}
\end{itemize}

%
\subsection{Category C}
\begin{itemize}
\item{\textbf{CA-AstroPh} is a an academic collaboration network for Astro Physics from Jan. 1993 - Apr. 2003.~\cite{snap}}
\item{\textbf{CA-CondMat} is an academic collaboration network for Condense Matter Physics.  Samples from 1999 (CondMat99), 2003 (CondMat03), and 2005 (CondMat05) were obtained from the University of Michigan.~\cite{umich}  A second sample from 2003 (CondMat03a) was obtained from Stanford University.~\cite{snap}}
\item{\textbf{CA-GrQc} is a an academic collaboration network for General Relativity and Quantum Cosmology from Jan. 1993 - Apr. 2003.~\cite{snap}}
\item{\textbf{CA-HepPh} is a an academic collaboration network for High Energy Physics - Phenomenology from Jan. 1993 - Apr. 2003.~\cite{snap}}
\item{\textbf{CA-HepTh} is a an academic collaboration network for High Energy Physics - Theory from Jan. 1993 - Apr. 2003.~\cite{snap}}
\item{\textbf{CA-NetSci} is a an academic collaboration network for Network Science from May 2006.}
\item{\textbf{Enron E-Mail} is an e-mail network from the Enron corporation made public by the Federal Energy Regulatory Commission during its investigation.~\cite{snap}}
\item{\textbf{URV E-Mail} is an e-mail network based on communications of members of the University Rovira i Virgili (Tarragona).~\cite{uvi}  It was extracted in 2003.}
\item{\textbf{YouTube} is a video-sharing website that allows users to establish friendship links.~\cite{Zafarani+Liu:2009}  The first sample (YouTube1) was extracted in Dec. 2008.  The second sample (YouTube2) was uploaded to ASU's Social Computing Data Repository in Aug. 2010.}
\end{itemize}

\begin{table}[ht]
     \begin{center}
        \includegraphics[width=.8\linewidth]{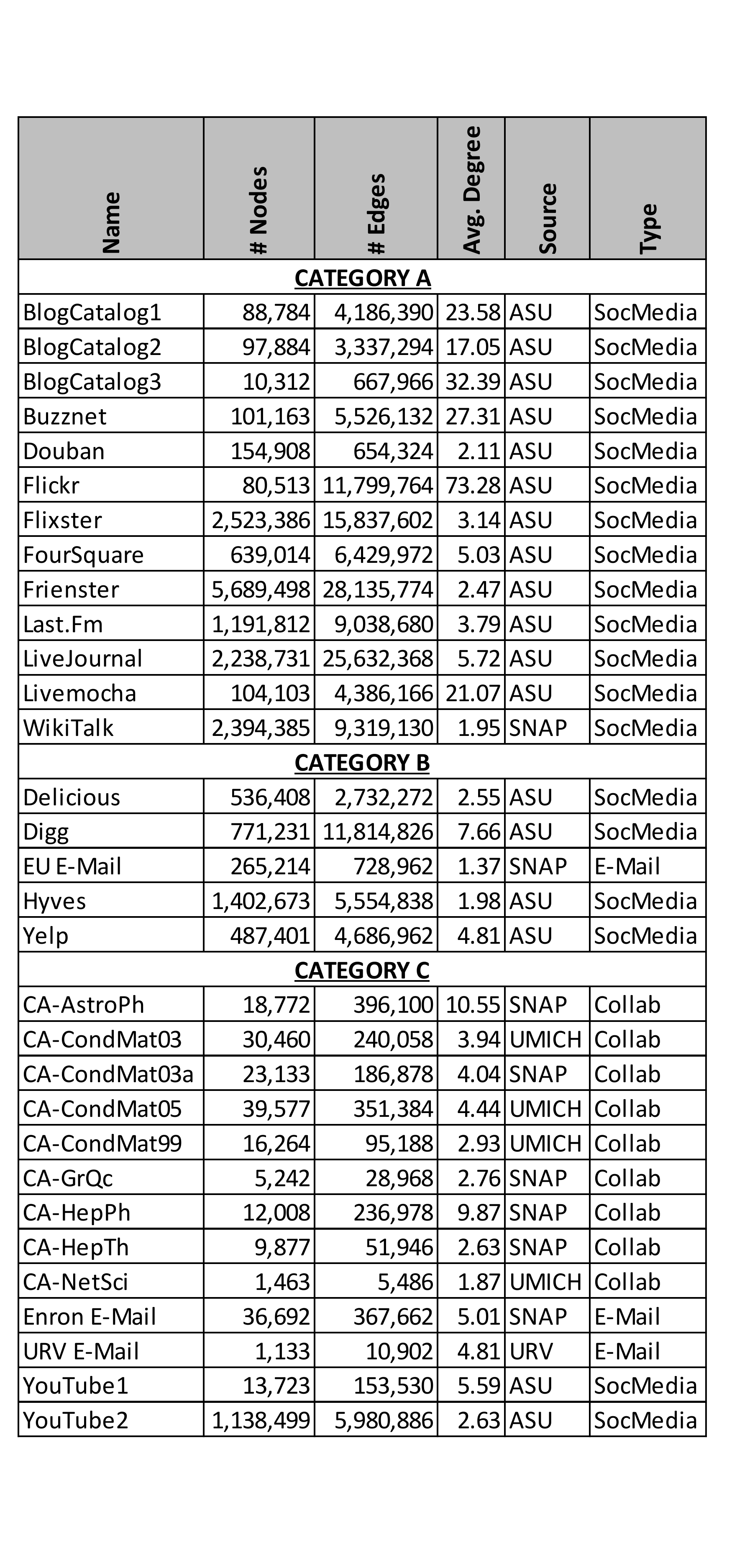}
    \end{center}
    \caption{Information on the networks in Categories A, B, and C.}
    \label{fig3}
\end{table}

\subsection{Runtime}
First, we examined the runtime of the algorithm (see Figure~\ref{figRt}).  Our experiments aligned well with our time complexity result (Proposition~\ref{tcomp}).  For example, a network extracted from the Dutch social-media site Hyves consisting of $1.4$ million nodes and $5.5$ million directed edges was processed by our algorithm in at most $12.2$ minutes.  The often-cited LiveJournal dataset consisting of $2.2$ million nodes and $25.6$ million directed edges was processed in no more than $66$ minutes - a short time for an NP-hard combinatorial problem on a large-sized input.

\begin{figure}[htbb]
    \begin{center}
        \includegraphics[width=.89\linewidth]{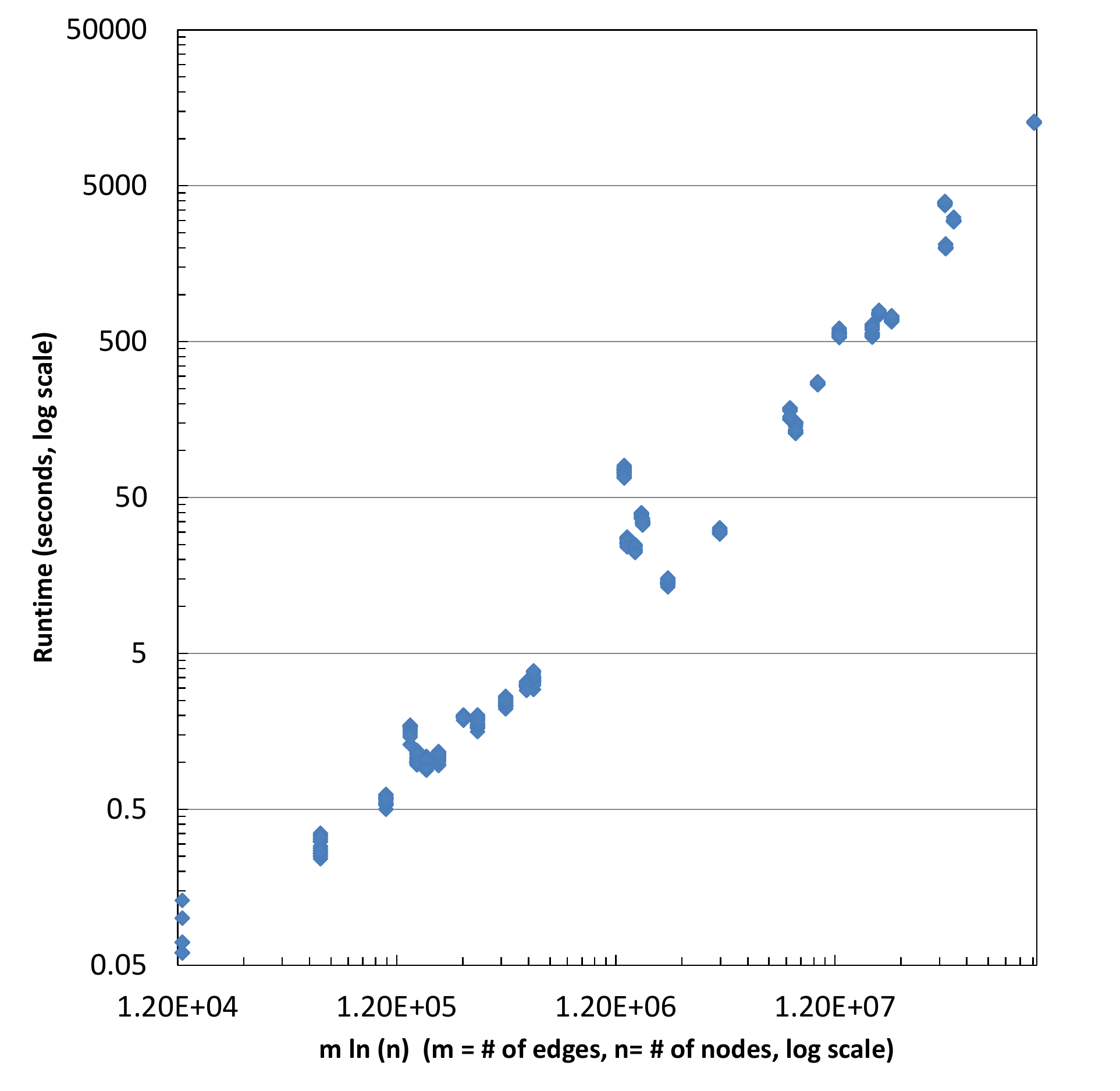}
    \end{center}
    \caption{$m \ln n$ vs. runtime in seconds (log scale, $m$ is number of edges, $n$ is number of nodes).  The relationship is linear with $R^2=0.9015$, $p=2.2 \cdot 10^{-16}$.}
    \label{figRt}
\end{figure}

\subsection{Seed Size}

For each network, we performed $10$ ``integer'' trials.  In these trials, we set $\theta(v_i)=\min(d^{in}_i,k)$ where $k$ was kept constant among all vertices for each trial and set at an integer in the interval $[1,10]$.  We evaluated the ability of a network to promote spreading under the tipping model based on the size of the set of nodes returned by our algorithm (as a percentage of total nodes).  For purposes of discussion, we have grouped our networks into three categories based on results (Figure~\ref{figFirst} and Table~\ref{figX}).  In general, online social networks had the smallest seed sets - $13$ networks of this type had an average seed set size less than $2\%$ of the population.  We also noticed, that for most networks, there was a linear realtion between threshold value and seed size.

\begin{figure}
    \begin{center}
        \includegraphics[width=1\linewidth]{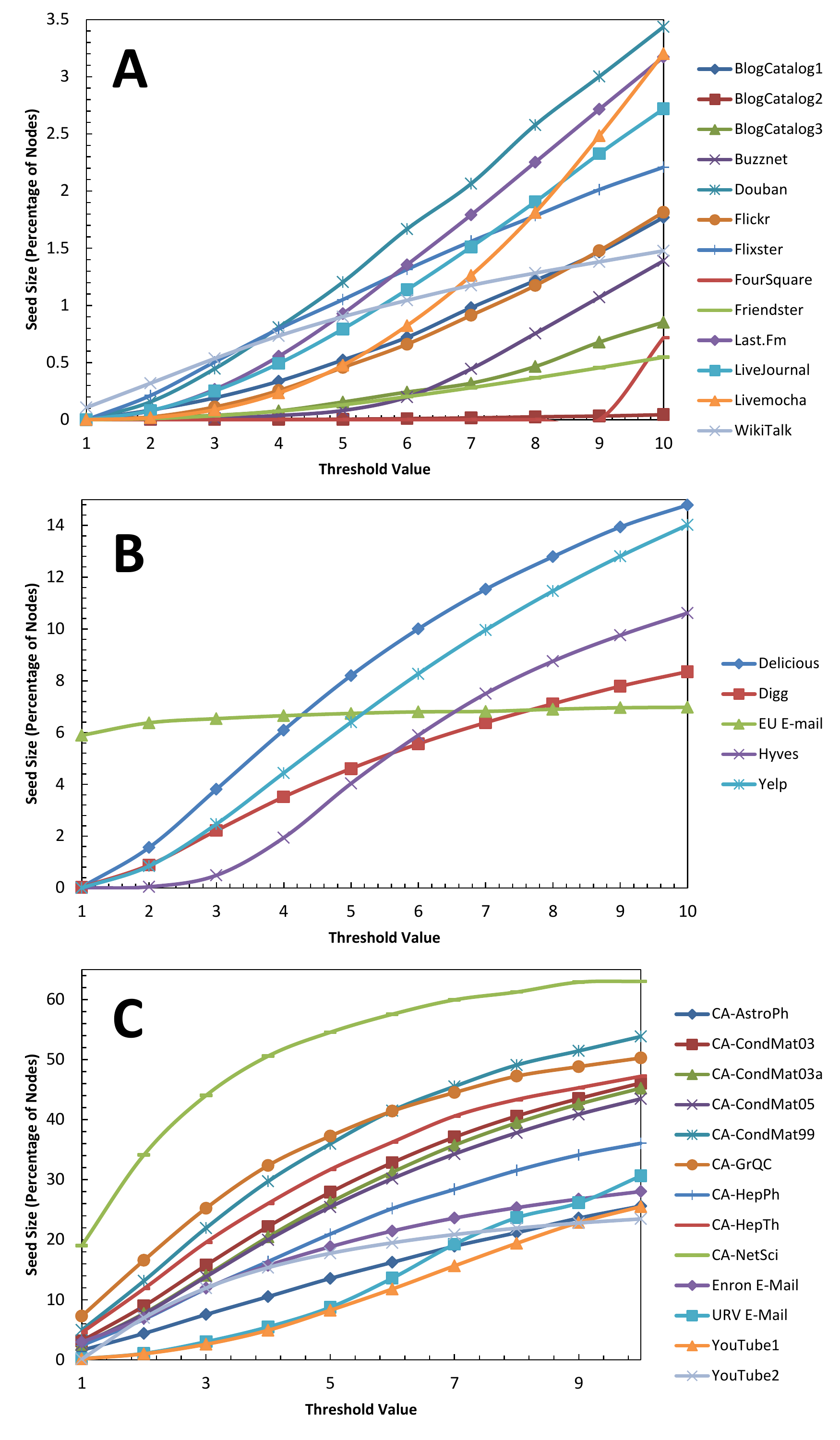}
    \end{center}
       \caption{Threshold value (assigned as an integer in the interval $[1,10]$) vs. size of initial seed set as returned by our algorithm in our three identified categories of networks (categories A-C are depicted in panels A-C respectively).  Average seed sizes were under $2\%$ for Categorty A, $2-10\%$ for Category B and over $10\%$ for Category C.  The relationship, in general, was linear for categories A and B and lograthimic for C.  CA-NetSci had the largest Louvain Modularity and clustering coefficient of all the networks.  This likely explains why that particular network seems to inhibit spreading.}
    \label{figFirst}
\end{figure}

Category A can be thought of as social networks highly susceptible to influence - as a very small fraction of individuals initially having a behavior can lead to adoption by the entire population.  In our ten trials, the average seed size was under $2\%$ for each of these $13$ networks.  All were extracted from social media websites.  For some of the lower threshold levels, the size of the set of seed nodes was particularly small.  For a threshold of three we had $11$ of the Category A networks with a seed size less than $0.5\%$ of the population.  For a threshold of four, we had nine networks meeting that criteria.

Networks in Category B are susceptible to influence with a relatively small set of initial nodes - but not to the extent of those in Category A.  They had an average initial seed size greater than $2\%$ but less than $10\%$. Members in this group included two general purpose social media networks, two specialty social media networks, and an e-mail network.

Category C consisted of networks that seemed to hamper diffusion in the tipping model, having an average initial seed size greater than $10\%$.  This category included all of the academic collaboration networks, two of the email networks, and two networks derived from friendship links on YouTube.

\subsection{Seed Size as a Function of Community Structure}

In this section, we view the results of our heuristic algorithm as a measurement of how well a given network promotes spreading.  Here, we use this measurement to gain insight into which structural aspects make a  network more likely to be ``tipped.''  We compared our results with two network-wide measures characterizing community structure.  First, clustering coefficient ($C$) is defined for a node as the fraction of neighbor pairs that share an edge - making a triangle.  For the undirected case, we define this concept formally below.

\begin{definition}[Clustering Coefficient]
Let $ r $ be the number of edges between nodes with which $ v_i $ has an edge and $d_i$ be the degree of $v_i$.  The \textbf{clustering coefficient}, $ C_i = \dfrac {2r}  {d_i(d_i - 1)} $.
\end{definition}

Intuitively, a node with high $C_i$ tends to have more pairs of friends that are also mutual friends.  We use the average clustering coefficient as a network-wide measure of this local property.

Second, we consider modularity ($M$) defined by Newman and Girvan.~\cite{newman04}.  For a partition of a network, $M$ is a real number in $[-1,1]$ that measures the density of edges within partitions compared to the density of edges between partitions.  We present a formal definition for an undirected network below.

\begin{definition}[Modularity~\cite{newman04}]
\textbf{Modularity}, $ M = \dfrac 1 {2m} \sum_{i,j \in V} [1 - \dfrac {d_i d_j} {2m}] \delta (c_i, c_j) $, where $m$ is the number of undirected edges, $d_i$ is node degree, $ c_i $ is the community to which $ v_i $ belongs and $ \delta (x, y) = 1 $ if $ x = y $ and $ 0 $ otherwise.
\end{definition}

The modularity of an optimal network partition can be used to measure the quality of its community structure.  Though modularity-maximization is NP-hard, the approximation algorithm of Blondel et al.~\cite{blondel08} (a.k.a. the ``Louvain algorithm'') has been shown to produce near-optimal partitions.\footnote{Louvain modularity was computed using the implementation available from CRANS at  http://perso.crans.org/aynaud/communities/.}  We call the modularity associated with this algorithm the ``Louvain modularity.''  Unlike the $C$, which describes local properties, $M$ is descriptive of the community level.  For the $31$ networks we considered, $M$ and $C$  appear uncorrelated ($R^2 = 0.0538$, $p=0.2092$).

We plotted the initial seed set size ($S$) (from our algorithm - averaged over the $10$ threshold settings) as a function of $M$ and $C$ (Figure~\ref{main-fig}a) and uncovered a correlation (planar fit, $R^2=0.8666$, $p=5.666 \cdot 10^{-13}$, see Figure~\ref{main-fig} A).  The majority of networks in Category C (less susceptible to spreading) were characterized by relatively large $M$ and $C$ (Category C includes the top nine networks w.r.t. $C$ and top five w.r.t. $M$).  Hence, networks with dense, segregated, and close-knit communities (large $M$ and $C$) suppress spreading.  Likewise, those with low $M$ and $C$ tended to promote spreading.  Also, we note that there were networks that promoted spreading with dense and segregated communities, yet were less clustered (i.e. Category A networks Friendster and LiveJournal both have $M\geq 0.65$ and $C \leq 0.13$).  Further, some networks with a moderately large clustering coefficient were also in Category A (two networks extracted from BlogCatalog had $C\geq 0.46$) but had a relatively less dense community structure (for those two networks $M \leq 0.33$).

\begin{figure}
    \begin{center}
        \includegraphics[width=1\linewidth]{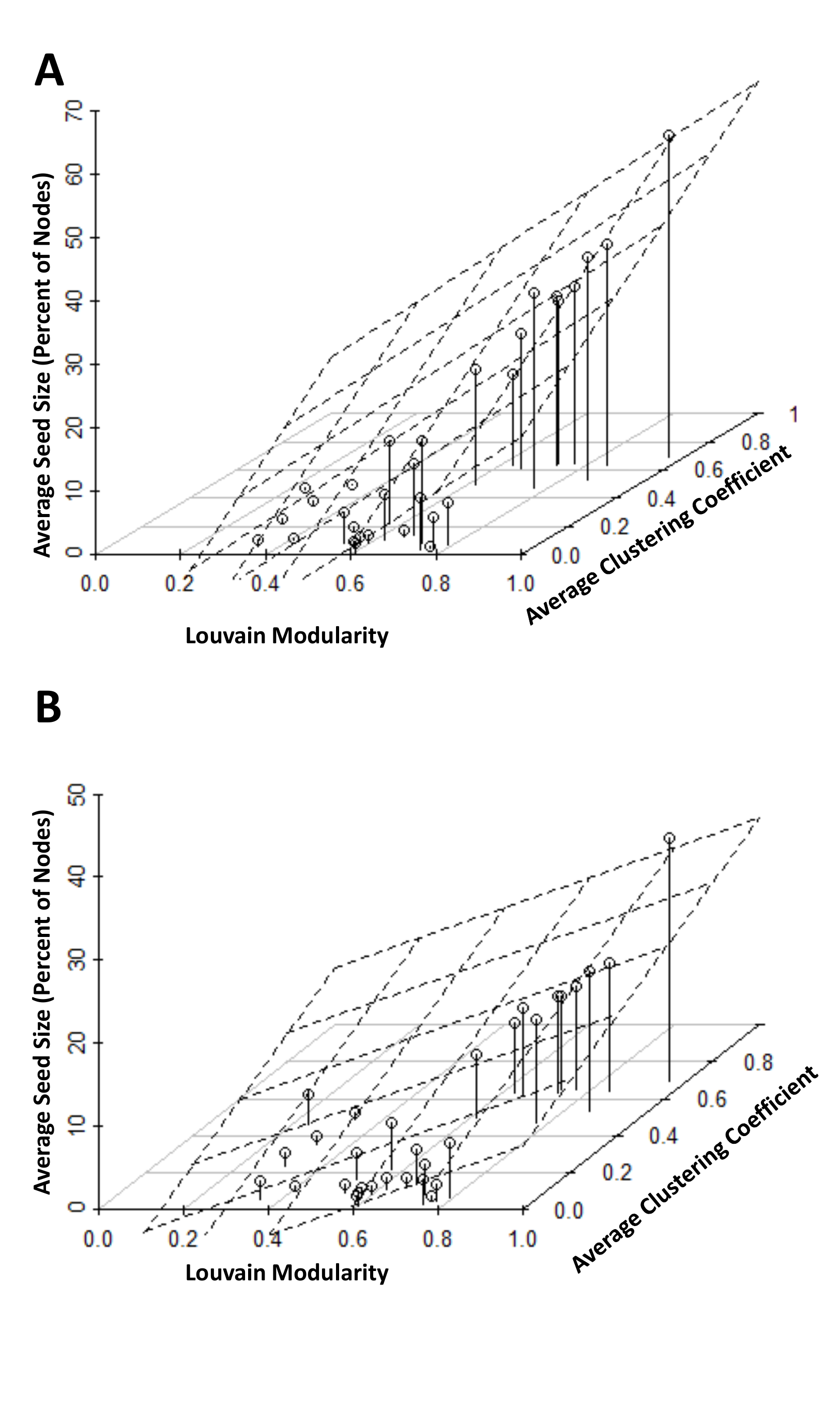}
    \end{center}
    \caption{\textbf{(A)} Louvain modularity ($M$) and average clustering coefficient ($C$) vs. the average seed size ($S$).  The planar fit depicted is $S=43.374 \cdot M +  33.794 \cdot C - 24.940$ with $R^2=0.8666$, $p=5.666 \cdot 10^{-13}$.  \textbf{(B)} Same plot at (A) except the averages are over the 12 percentage-based threshold values.  The planar fit depicted is $S=18.105 \cdot M + 17.257 \cdot C - 10.388$ with $R^2=0.816$, $p=5.117 \cdot 10^{-11}$.}
    \label{main-fig}
\end{figure}

We also studied the effects on spreading when the threshold values would be assigned as a certain fraction of the node's in-degree.~\cite{jy05,wattsDodds07}  This results in heterogeneous $\theta_i$'s for the nodes.  We performed $12$ trials for each network.  Thresholds for each trial were based on the product of in-degree and a fraction in the interval $[0.05,0.60]$ (multiples of $0.05$).  The results (Figure~\ref{figFrac} and Table~\ref{figX}) were analogous to our integer tests.  We also compared the averages over these trials with $M$ and $C$ and obtained similar results as with the other trials (Figure~\ref{main-fig} B).

\begin{figure}
    \begin{center}
        \includegraphics[width=1\linewidth]{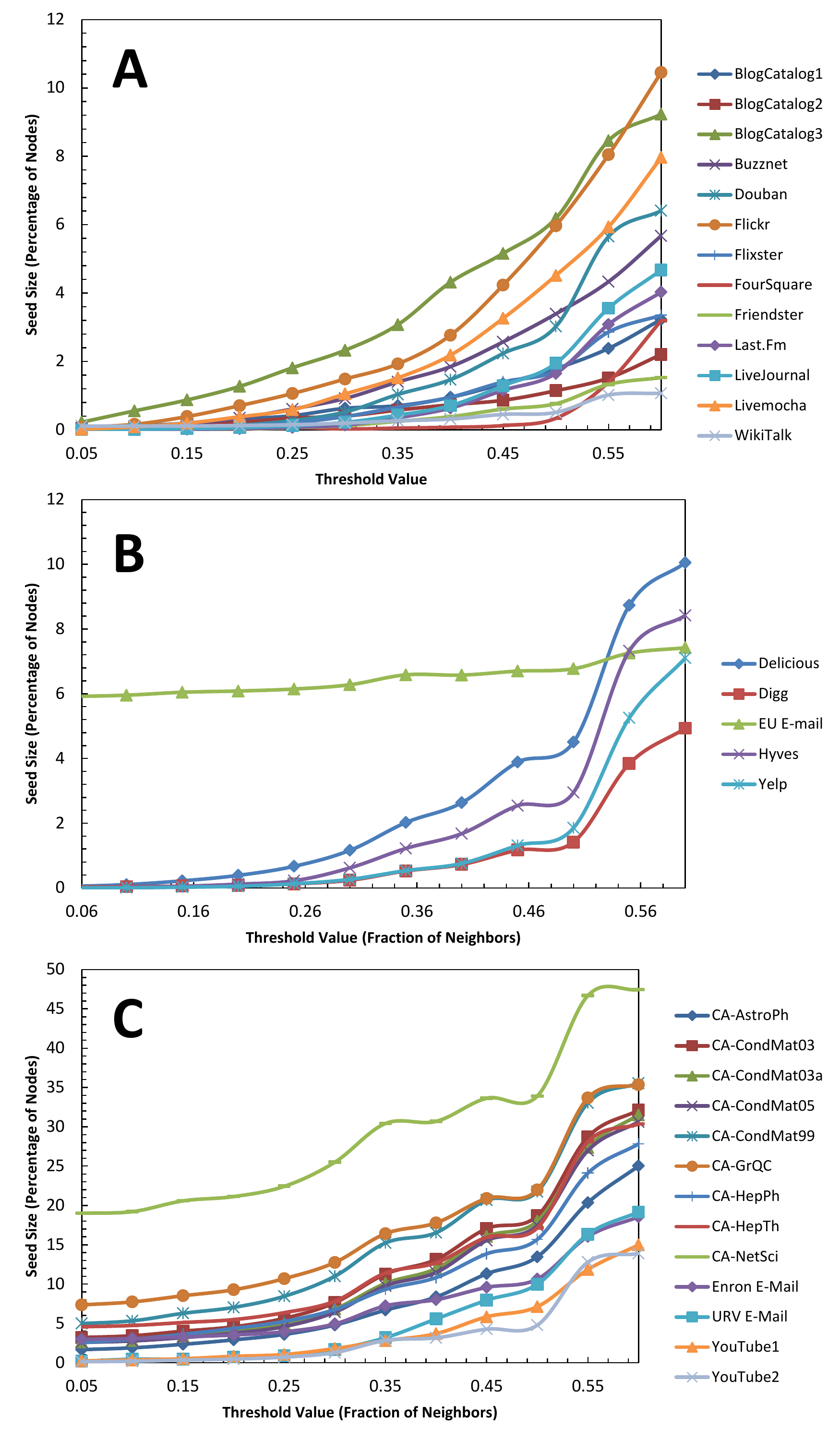}
    \end{center}
   \caption{Threshold value (assigned as a fraction of node in-degree as a multiple of $0.05$ in the interval $[0.05,0.60]$) vs. size of initial seed set as returned by our algorithm in our three identified categories of networks (categories A-C are depicted in panels A-C respectively, categories are the same as in Figure 1). Average seed sizes were under $5\%$ for Categorty A, $1-7\%$ for Category B and over $3\%$ for Category C.  In general, the relationship between threshold and initial seed size for networks in all categories was exponential.}
    \label{figFrac}
\end{figure}

\begin{table}[ht]
     \begin{center}
        \includegraphics[width=.8\linewidth]{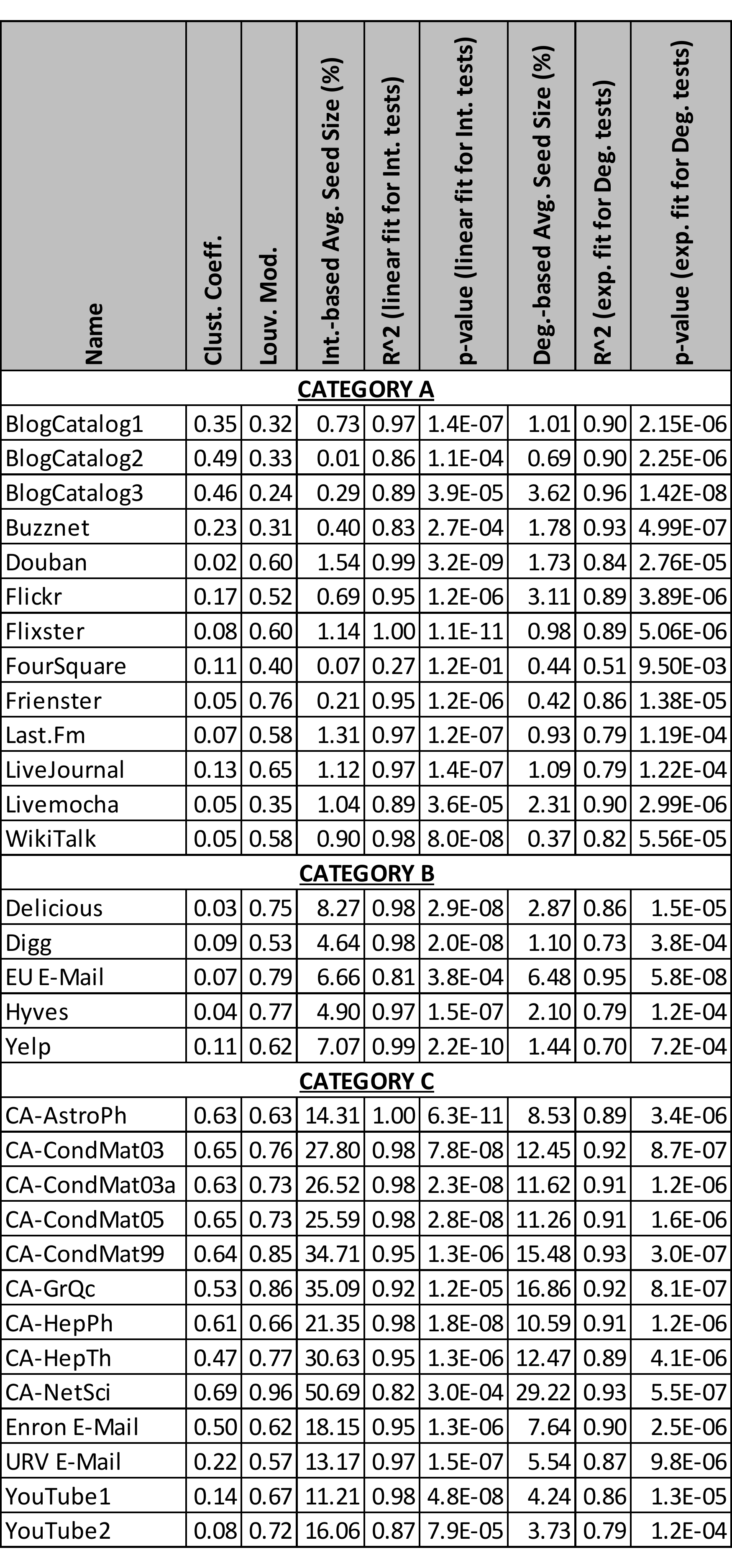}
    \end{center}
    \caption{Regression analysis and network-wide measures for the networks in Categories A, B, and C.}
    \label{figX}
\end{table}

\section{Related Work}
\label{rw-sec}
Tipping models first became popular by the works of \cite{Gran78} and \cite{Schelling78} where it was presented primarily in a social context.  Since then, several variants have been introduced in the literature including the non-deterministic version of \cite{kleinberg} (described later in this section) and a generalized version of \cite{jy05}.  In this paper we focused on the deterministic version.  In \cite{wattsDodds07}, the authors look at deterministic tipping where each node is activated upon a percentage of neighbors being activated.  Dryer and Roberts \cite{Dreyer09} introduce the MIN-SEED problem, study its complexity, and describe several of its properties w.r.t. certain special cases of graphs/networks.  The hardness of approximation for this problem is described in \cite{chen09siam}.  The work of \cite{benzwi09} presents an algorithm for target-set selection whose complexity is determined by the tree-width of the graph - though it provides no experiments or evidence that the algorithm can scale for large datasets.  The recent work of \cite{reichman12} prove a non-trivial upper bound on the smallest seed set.

Our algorithm is based on the idea of shell-decomposition that currently is prevalent in physics literature.  In this process, which was introduced in \cite{Seidman83}, vertices (and their adjacent edges) are iteratively pruned from the network until a network ``core'' is produced.  In the most common case, for some value $k$, nodes whose degree is less than $k$ are pruned (in order of degree) until no more nodes can be removed.  This process was used to model the Internet in \cite{ShaiCarmi07032007} and find key spreaders under the SIR epidemic model in \cite{InfluentialSpreaders_2010}.  More recently, a ``heterogeneous'' version of decomposition was introduced in \cite{baxter11} - in which each node is pruned according to a certain parameter - and the process is studied in that work based on a probability distribution of nodes with certain values for this parameter.

\subsection{Notes on Non-Deterministic Tipping}
We also note that an alternate version of the model where the thresholds are assigned randomly has inspired approximation schemes for the corresponding version of the seed set problem.\cite{kleinberg,leskovec07,chen10}  Work in this area focused on finding a seed set of a certain size that maximizes of the expected number of adopters.  The main finding by Kempe et al., the classic work for this model, was to prove that the expected number of adopters was submodular - which allowed for a greedy approximation scheme.  In this algorithm, at each iteration, the node which allows for the greatest increase in the expected number of adopters is selected.  The approximation guarantee obtained (less than $0.63$ of optimal) is contingent upon an approximation guarantee for determining the expected number of adopters - which was later proved to be $\#P$-hard.~\cite{chen10}  Though finding a such a guarantee is still an open question, work on counting-complexity problems such as that of Dan Roth~\cite{roth96} indicate that a non-trivial approximation ratio is unlikely.  Further, the simulation operation is often expensive - causing the overall time complexity to be $O(x \cdot n^2)$ where $x$ is the number of runs per simulation and $n$ is the number of nodes (typically, $x>n$).  In order to avoid simulation, various heuristics have been proposed, but these typically rely on the computation of geodesics - an $O(n^3)$ operation - which is also more expensive than our approach.

Additionally, the approximation argument for the non-deterministic case does not directly apply to the original (deterministic) model presented in this paper.  A simple counter-example shows that sub-modularity does not hold here. Sub-modularity (diminishing returns) is the property leveraged by Kempe et al. in their approximation result.

\subsection{Note on an Upper Bound of the Initial Seed Set}

Very recently, we were made aware of research by Daniel Reichman that proves an upper bound on the minimal size of a seed set for the special case of undirected networks with homogeneous threshold values.~\cite{reichman12}  The proof is constructive and yields an algorithm that mirrors our approach (although Reicshman's algorithm applies only to that special case).  We note that our work and the work of Reichman were developed independently.  We also note that Reichman performs no experimental evaluation of the algorithm.

Given undirected network $G$ where each node $v_i$ has degree $d_i$ and the threshold value for all nodes is $k$, Reichman proves that the size of the minimal seed set can be bounded by $\sum_i \min\{1, \frac{k}{d_i+1}\}$.  For our integer tests, we compared our results to Reichman's bound.  Our seed sets were considerably smaller - often by an order of magnitude or more.  See Figure~\ref{fig4a} for details.

\begin{figure}
    \begin{center}
        \includegraphics[width=1\linewidth]{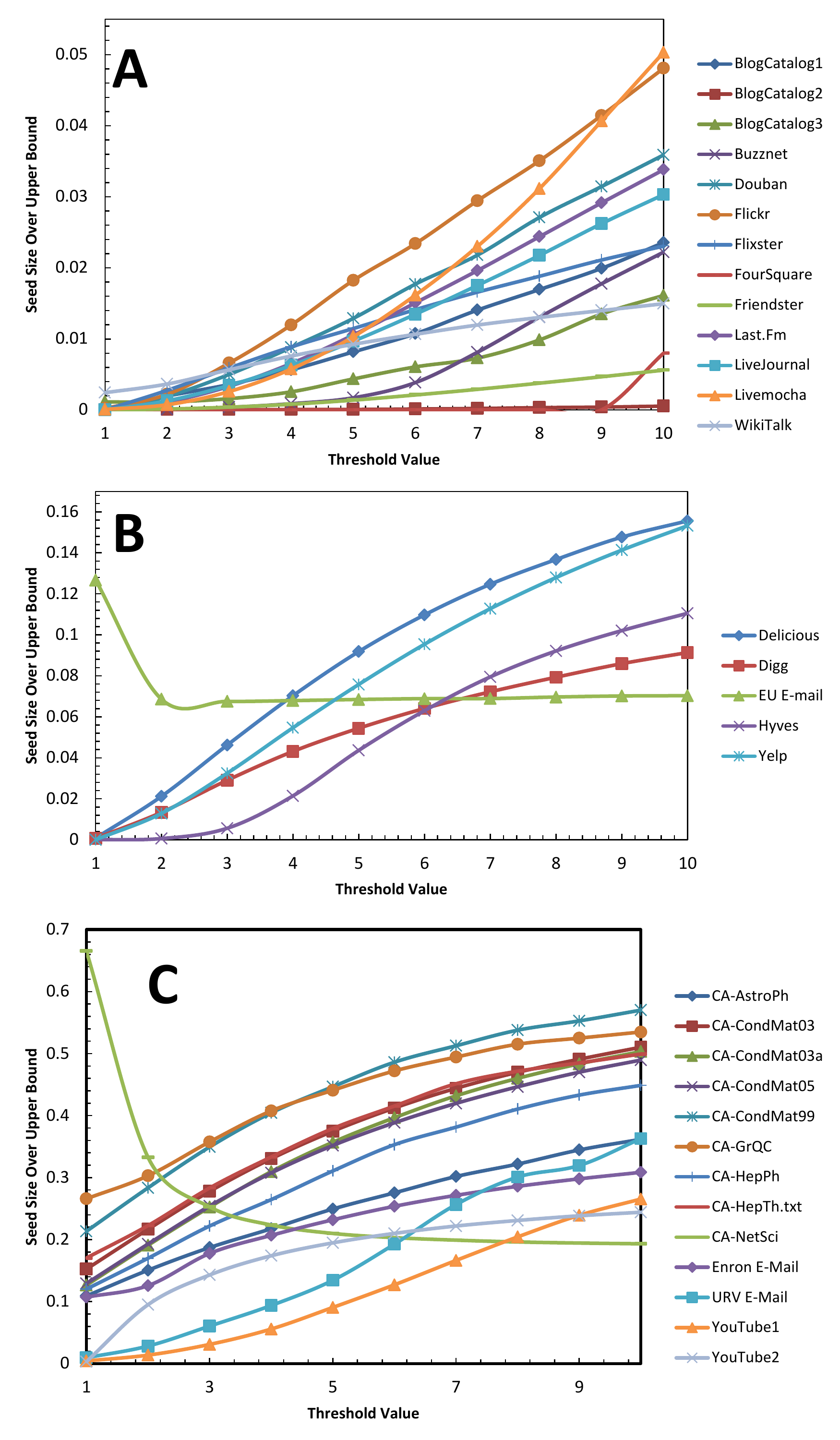}
    \end{center}
    \caption{Integer threshold values vs. the seed size divided by Reichman's upper bound~\cite{reichman12} the three categories of networks (categories A-C are depicted in panels A-C respectively).  Note that in nearly every trial, our algorithm produced an initial seed set significantly smaller than the bound - in many cases by an order of magnitude or more.}
    \label{fig4a}
\end{figure}

\section{Conclusion}

As recent empirical work on tipping indicates that it can occur in real social networks,\cite{centola10,zhang11} our results are encouraging for viral marketers.  Even if we assume relatively large threshold values, small initial seed sizes can often be found using our fast algorithm - even for large datasets.  For example, with the FourSquare online social network, under majority threshold ($50\%$ of incoming neighbors previously adopted), a viral marketeer could expect a $297$-fold return on investment.  As results of this type seem to hold for many online social networks, our algorithm seems to hold promise for those wishing to ``go viral.''

\section*{Acknowledgments}
We would like to thank Gaylen Wong (USMA) for his technical support.  Additionally, we would like to thank (in no particular order) Albert-L\'{a}szl\'{o} Barab\'{a}si (NEU), Sameet Sreenivasan (RPI), Boleslaw Szymanski (RPI), John James (USMA), and Chris Arney (USMA) for their discussions relating to this work.  Finally, we would also like to thank Megan Kearl, Javier Ivan Parra, and Reza Zafarani of ASU for their help with some of the datasets.  
The authors are supported under by the Army Research Office (project 2GDATXR042) and the Office of the Secretary of Defense (project F1AF262025G001).  The opinions in this paper are those of the authors and do not necessarily reflect the opinions of the funders, the U.S. Military Academy, or the U.S. Army.

\bibliographystyle{IEEEtranS}
\bibliography{ShakarianAndPaulo-tgtsInSocNws}

\end{document}